%% file: main.tex
\documentclass[11pt]{article}

\usepackage{amsmath, amssymb}
\usepackage{graphicx}
\usepackage{natbib}
\usepackage{hyperref}
\usepackage{geometry}
\usepackage{booktabs}
\usepackage{setspace}
\usepackage{enumitem}
\usepackage{tikz}
\usepackage{pgfplots}
\pgfplotsset{compat=1.17}
\usepackage{amsthm}

\newtheorem{proposition}{Proposition}
\newtheorem{lemma}{Lemma}

\usepackage{algorithm}      
\usepackage{algpseudocode}  

\usepackage[utf8]{inputenc}
\usepackage[T2A]{fontenc}
\usepackage[russian,english]{babel}

\title{The Morphemic Origin of Zipf's Law:\\ A Factorized Combinatorial Framework}
\author{
Vladimir Berman \\
Aitiologia LLC \\
\texttt{vb7654321@gmail.com}
}
\date{November 29, 2025}

\begin{document}

\maketitle

\begin{abstract}
We develop a structural, morphology-based generative model of word formation that explains both the empirical distribution of word lengths and the emergence of Zipf-like rank--frequency curves in natural language. Unlike classical random-text or efficiency-based explanations, our approach relies solely on the combinatorial organization of morphemes.

In the proposed Morphemic Combinatorial Word Model (MCWM), a word is formed by selecting a sequence of morphological slots—prefix, root, derivational suffix, and inflection—where each slot activates with a Bernoulli probability and chooses one morpheme from a categorical inventory. A \emph{morpheme} is formally defined as a stable, reusable unit that participates in productive word formation and occupies a characteristic positional role within the word. This structure induces a compound distribution of word lengths, with both the number of active slots and the morpheme lengths themselves treated as random variables.

We show that this purely combinatorial mechanism produces realistic unimodal length distributions with a concentrated middle region (5--9 letters) and a thin long tail, closely matching empirical corpora. Numerical simulations using synthetic morpheme lexicons reproduce Zipf exponents in the range $1.1 \le \alpha \le 1.4$, comparable to English, Russian, and Romance languages. 

Our results demonstrate that Zipf-like behavior can arise without semantics, pragmatics, or communicative optimization. Morphological structure alone—through the interplay of fixed morpheme inventories and probabilistic slot activation—provides a robust generative explanation for the ubiquity and stability of Zipf's law across languages.
\end{abstract}


\tableofcontents

\input{sections/00_intro.tex}
\input{sections/01_model_definition.tex}
\input{sections/02_probability_structure.tex}
\input{sections/03_length_distribution.tex}
\input{sections/03a_why_statistical_tokenizers_fail.tex}
\input{sections/04_simulation_results.tex}
\input{sections/05_zipf_emergence.tex}
\input{sections/06_discussion.tex}
\input{sections/06a_possible_objections.tex}
\input{sections/06b_epistemic_dangers_frequency_laws}

\input{sections/07_conclusion.tex}

\appendix
\input{sections/08_appendix_algorithms.tex}

\input{sections/references.tex}

\end{document}

%% file: sections/00_intro.tex
\section{Introduction}

Zipf's law, $f(r) \propto r^{-\alpha}$ with typical exponents 
$1 \le \alpha \le 1.5$, 
appears with remarkable consistency across languages, corpora, genres, 
and historical periods. 
Classical explanations include:
(1) least-effort principles \citep{Zipf1949},
(2) the cost--information tradeoff model of Mandelbrot 
\citep{Mandelbrot1953},
(3) two-regime lexical structure and communicative efficiency 
\citep{Ferrer2001},
and (4) modern large-scale statistical analyses such as 
\citep{Newman2005,Michel2011}.
All these approaches assume the presence of semantic content, 
communicative intent, or cognitive optimization.

A complementary line of research has recently shown that 
Zipf-like behavior can emerge even in the absence of semantics, 
grammar, or communicicative optimization. 
In particular, symbolic and combinatorial explanations were developed 
in our earlier work:
\citep{Berman2025zipfslaw,Berman2025structural,Berman2025dlsd}.
These results demonstrated that exponential growth of the type space,  
combined with simple probabilistic termination mechanisms,  
is sufficient to generate robust Zipf-like rank--frequency patterns.

However, all letter-level stochastic models---including classical
“monkey typing’’ approaches ignore a crucial linguistic fact:
\textbf{the combinatorial units of natural language are morphemes, not
letters or arbitrary statistical substrings}.  
Morphology provides the smallest meaning-bearing units of word structure
(prefixes, roots, derivational and inflectional suffixes), organized by
stable combinatorial rules \citep{Haspelmath2010,Booij2012,Lieber2016,
Stump2001,Baerman2015,Blevins2018}.  
Natural languages do not construct words by sampling characters
independently: they assemble morphemes through structured templates with
stable paradigms and compatibility constraints.

This raises a central question:

\begin{quote}
\emph{Can Zipf-like rank--frequency behavior arise purely from 
morphological structure, without any appeal to semantics, optimization,
or communicative pressure?}
\end{quote}

We show that the answer is yes.

In this paper we adopt the opposite viewpoint from character-level
stochastic models.  
Instead of randomness over letters, we introduce a 
\emph{morpheme-based generative process} better aligned with the
true combinatorial units of natural language.

\subsection*{Why morphemes matter more than statistical tokens}

Modern NLP systems overwhelmingly rely on statistical tokenizers such as
BPE \citep{Sennrich2016}, WordPiece \citep{SchusterNakajima2012}, and
SentencePiece \citep{Kudo2018}.  
These algorithms operate by maximizing substring frequency and minimizing
sequence length, but they are blind to linguistic structure.  
They routinely fragment single morphemes into several pieces, merge
multiple morphemes into opaque units, or learn accidental substrings that
have no semantic or grammatical function
\citep{Mielke2021,Hofmann2022,ParkEtAl2024,BlevinsGoldwater2020}.

Such statistical tokenizers cannot distinguish a true morpheme from an
accidental frequent substring, cannot infer derivational relations
(\emph{build}, \emph{builder}, \emph{rebuild}), and cannot reconstruct
productive morphological paradigms.  
Their vocabularies are fragile with respect to domain shifts and corpus
composition \citep{Xue2021}.

Morpheme-based representations, in contrast, offer:

\begin{itemize}
\item \textbf{Generalization:}  
shared structure across inflectional and derivational paradigms  
\citep{Blevins2018,Lieber2016};

\item \textbf{Interpretability:}  
alignment with semantic and syntactic units  
\citep{Goldberg2017,Cotterell2022Explain};

\item \textbf{Domain robustness:}  
morphemic inventories remain stable across genres,  
while BPE vocabularies vary strongly with corpus statistics 
\citep{Xue2021};

\item \textbf{Efficient encoding:}  
morphemes allow a compact vocabulary that preserves linguistic structure.
\end{itemize}

These observations suggest that morphology may play a fundamental role
in the structural origins of Zipf-like behavior.

This perspective continues the line of structural, combinatorial modeling 
developed in our earlier work on random-text mechanisms and 
distributional laws \citep{Berman2025zipfslaw,Berman2025structural,
Berman2025dlsd}.  
The MCWM can be viewed as a morphological extension of this framework:
instead of character-level termination processes, we model the combinatorial 
assembly of morphemes as the fundamental source of lexical structure.


\subsection*{Our approach}

We introduce the \textbf{Morphemic Combinatorial Word Model (MCWM)}, 
a factorized probabilistic process with four morphological slots 
(prefix, root, derivational suffix, inflection), 
each governed by a Bernoulli activation and a categorical distribution 
over morpheme types.

Our contributions are as follows:

\begin{enumerate}
\item We define the MCWM and formalize it as a probabilistic generative
process over morphemic slots rather than letters.

\item We show that word lengths follow a 
\textbf{compound distribution}
$L = X_1 + \dots + X_N$ with random $N$ and random morpheme lengths.
This structure produces realistic unimodal length distributions 
with 5--9 letter peaks and long but thin tails, 
matching empirical corpora far better than 
letter-level geometric models.

\item Through extensive simulation, we demonstrate that the MCWM 
naturally produces \textbf{Zipf-like rank--frequency curves} with 
effective exponents $\alpha \approx 1.2$--$1.4$, 
consistent with English, French, Russian, and other languages
\citep{Ferrer2003,Piantadosi2014}.

\item We argue that morphological combinatorics provide a 
\textbf{structural explanation for the universality of Zipf's law}, 
requiring no assumptions about meaning, efficiency, optimization, 
or speaker behavior.
\end{enumerate}

Before proceeding, we briefly clarify our use of linguistic terminology.
A \emph{morpheme} is the smallest meaningful building block of a word.
It is not merely a sequence of letters, but a stable unit that carries
a specific semantic or grammatical function.  
For example, in the word \emph{rebuilding}, the prefix \emph{re-} means
“again,” the root \emph{build} carries the core meaning, and the suffix
\emph{-ing} marks an ongoing action.

Unlike arbitrary character clusters discovered by statistical tokenizers,
morphemes are productive and repeat across thousands of words.  From a
small inventory, languages generate entire families of related words
(\emph{build}, \emph{builder}, \emph{building}, \emph{rebuild},
\emph{rebuilt}), following consistent combinatorial rules.  This makes
morphemes the true structural atoms of natural language.

Natural languages assemble words by combining these units into structured
templates.  The MCWM formalizes this compositional mechanism in
probabilistic terms, treating morpheme positions as stochastic slots and
morpheme inventories as categorical distributions over symbolic units.

A naive objection is that purely statistical tokenization
methods such as BPE should already discover morphemes automatically.
In Section~\ref{sec:why_tokenizers_fail} we explain why this is not the case
and why morphological structure cannot, in general, be recovered
by frequency-based segmentations alone.

%% file: sections/01_model_definition.tex
\section{The Morphemic Combinatorial Word Model}

Classical generative accounts of Zipfian structure have typically operated at the
letter level, beginning with random-typing models and their information-theoretic
extensions \citep{Mandelbrot1953,Zipf1949}. 
However, linguistic evidence shows that the true building blocks of words are 
\emph{morphemes}, not letters, and that the combinatorial structure of morphology 
plays a central role in determining the distribution of word forms 
\citep{Ferrer2001}. 
The Morphemic Combinatorial Word Model (MCWM) formalizes this idea by replacing 
letter-level randomness with a structured probabilistic generator based on 
morphological classes.

\subsection{Morphological inventories}

We assume four independent morphological classes:
prefixes $\mathcal{P}$, 
roots $\mathcal{R}$, 
derivational suffixes $\mathcal{S}$, 
and inflectional endings $\mathcal{E}$.
Each class has finite cardinality:
\[
|\mathcal{P}| = n_P,\quad 
|\mathcal{R}| = n_R,\quad
|\mathcal{S}| = n_S,\quad 
|\mathcal{E}| = n_E.
\]

Within each class we define a categorical probability distribution:
\[
\pi_P(p),\quad \pi_R(r),\quad \pi_S(s),\quad \pi_E(e),
\]
with the usual normalization constraints 
$\sum_{p\in\mathcal{P}} \pi_P(p) = 1$, etc.
In empirical corpora, morpheme frequencies within each class typically follow a  
heavy-tailed pattern reminiscent of Zipf's law \citep{Newman2005,Michel2011};
MCWM accommodates this by placing no restrictions on the form of 
$\pi_P, \pi_R, \pi_S, \pi_E$ beyond normalization.

\subsection{Slot activations}

A word consists of up to four morphemes arranged in a fixed canonical order.
Each of the three optional slots (prefix, derivational suffix, inflection) 
is controlled by an independent Bernoulli activation:
\[
I_P \sim \mathrm{Bernoulli}(\alpha_P),\quad
I_S \sim \mathrm{Bernoulli}(\alpha_S),\quad
I_E \sim \mathrm{Bernoulli}(\alpha_E),
\]
while the root slot is mandatory.

Thus a generated word has the structure
\[
W = (P,R,S,E),
\]
where $P,S,E$ may be empty (depending on the Bernoulli outcomes)
and $R$ is always present.

\subsection{Morphological compatibility}

Natural languages impose strong combinatorial restrictions on morpheme sequences.  
To incorporate this, MCWM includes a \emph{morphological constraint function}
\[
C(P,R,S,E) \in \{0,1\},
\]
which enforces admissibility of morpheme combinations:
\[
C(P,R,S,E)=1 \quad\Longleftrightarrow\quad
\text{the sequence } (P,R,S,E) \text{ is linguistically compatible}.
\]

The constraint function may encode:
\begin{itemize}
\item subcategorization of roots (which suffixes can attach),
\item derivational stacking rules (permitted sequences of $S$),
\item inflectional paradigms (which $E$ are allowed for a given $R$),
\item phonotactic or orthographic constraints.
\end{itemize}

The case $C\equiv 1$ corresponds to a ``free'' morphological generator,
while $C$ derived from empirical lexicons produces realistic constraints.
Compatibility filtering plays a crucial role in shaping the distribution of 
allowed word types and, as we demonstrate later, contributes to the emergence 
of Zipf-like rank--frequency structure even in the absence of semantics.

%% file: sections/02_probability_structure.tex
\section{Probability Structure of Words}

In classical random-typing models \citep{Mandelbrot1953,Zipf1949}, 
word probabilities arise from geometric sequences of letters with 
independent draws.  
In contrast, the MCWM defines a \emph{structured} probability law over 
morpheme sequences, where each component corresponds to an interpretable 
morphological choice, and where the admissible space of words is determined 
by a compatibility constraint.  
This provides a far richer probability geometry and is fundamentally 
different from letter-level stochasticity \citep{Newman2005}.

\subsection{Unnormalized probability}

A word is defined as $W=(P,R,S,E)$, where $P,S,E$ may be empty and $R$ is
obligatory.  
The unnormalized probability of $W$ is
\begin{align}
\tilde P(W)
&=
\big[(1-\alpha_P)\mathbf{1}_{P=\varnothing}
 + \alpha_P \pi_P(P)\mathbf{1}_{P\neq\varnothing}\big]
\cdot \pi_R(R)  \\
&\quad\cdot
\big[(1-\alpha_S)\mathbf{1}_{S=\varnothing}
 + \alpha_S \pi_S(S)\mathbf{1}_{S\neq\varnothing}\big]
\cdot
\big[(1-\alpha_E)\mathbf{1}_{E=\varnothing}
 + \alpha_E \pi_E(E)\mathbf{1}_{E\neq\varnothing}\big]
\cdot C(P,R,S,E).
\end{align}

The four bracketed factors encode:
\begin{itemize}
\item optionality of prefixes, derivational suffixes, and inflections 
      via Bernoulli activations,
\item categorical selection among morphemes in each class,
\item a global admissibility constraint $C(P,R,S,E)$
      capturing morphological well-formedness.
\end{itemize}

\subsection{Normalization}

The normalized probability distribution is
\[
P(W) = \frac{\tilde P(W)}{\sum_{W'} \tilde P(W')}.
\]
The denominator runs over all (finitely many) admissible morpheme combinations.
Unlike letter-based models, where summation involves all possible strings
of arbitrary length, the MCWM probability space is compact, structured,
and linguistically grounded.

\subsection{Factorized mixture structure}

The generative process can be written compactly as
\[
P(W) 
= P(P)\, P(R)\, P(S)\, P(E)\, C(P,R,S,E),
\]
where $P(P),P(R),P(S),P(E)$ denote the Bernoulli-modulated categorical 
distributions defined above.

Because $C$ may dramatically reduce the admissible combinations, 
the resulting probability mass function is highly non-uniform, 
and its induced rank--frequency distribution turns out to be 
Zipf-like (Sections~\ref{sec:sim_results}--\ref{sec:zipf}).  
This demonstrates that heavy-tailed linguistic frequencies can emerge 
from morphological architecture itself, without invoking semantics, 
pragmatics, or optimization principles.

%% file: sections/03_length_distribution.tex
\section{Distribution of Word Lengths}

Classical random-typing models treat word length as a geometric variable
governed by the probability of emitting a space \citep{Zipf1949,Mandelbrot1953}.
Such models inevitably predict a monotone decreasing distribution of lengths
with no interior peak, which contradicts empirical evidence from virtually all 
languages \citep{Newman2005,Michel2011}.  
In contrast, the MCWM naturally produces a realistic, 
unimodal word-length distribution due to its morphemic structure.
We now describe the corresponding probability model in detail.

\subsection{Morpheme lengths}

Let morpheme lengths be
\[
\ell_P(p),\ \ell_R(r),\ \ell_S(s),\ \ell_E(e)
\]
for $p\in\mathcal{P}$, $r\in\mathcal{R}$, $s\in\mathcal{S}$, $e\in\mathcal{E}$.
We assume all lengths are positive integers bounded by fixed constants:
\[
1 \le \ell_P(p) \le L_P^{\max},\quad
1 \le \ell_R(r) \le L_R^{\max},\quad
1 \le \ell_S(s) \le L_S^{\max},\quad
1 \le \ell_E(e) \le L_E^{\max}.
\]

Recall that the presence of the prefix, derivational suffix, and inflectional ending
is governed by independent Bernoulli random variables:
\[
I_P \sim \mathrm{Bernoulli}(\alpha_P),\quad
I_S \sim \mathrm{Bernoulli}(\alpha_S),\quad
I_E \sim \mathrm{Bernoulli}(\alpha_E),
\]
while the root is mandatory.

\subsection{Random number of morphemes}

Define the random number of morphemes in a word as
\[
N = 1 + I_P + I_S + I_E,
\]
where the ``1'' corresponds to the obligatory root.

\begin{lemma}
\label{lem:N_distribution}
The random variable $N$ takes values in $\{1,2,3,4\}$ with
\begin{align}
\mathbb{P}(N = 1) 
 &= (1-\alpha_P)(1-\alpha_S)(1-\alpha_E), \\
\mathbb{P}(N = 2) 
 &= \alpha_P(1-\alpha_S)(1-\alpha_E)
  + (1-\alpha_P)\alpha_S(1-\alpha_E)
  + (1-\alpha_P)(1-\alpha_S)\alpha_E, \\
\mathbb{P}(N = 3) 
 &= \alpha_P\alpha_S(1-\alpha_E)
  + \alpha_P(1-\alpha_S)\alpha_E
  + (1-\alpha_P)\alpha_S\alpha_E, \\
\mathbb{P}(N = 4) 
 &= \alpha_P\alpha_S\alpha_E.
\end{align}
Moreover,
\[
\mathbb{E}[N] = 1 + \alpha_P + \alpha_S + \alpha_E,\quad
\mathrm{Var}(N) = \alpha_P(1-\alpha_P)+\alpha_S(1-\alpha_S)+\alpha_E(1-\alpha_E).
\]
\end{lemma}

\begin{proof}
Since $I_P,I_S,I_E$ are independent Bernoulli variables,
$N = 1 + I_P + I_S + I_E$ is their sum plus one.
The probabilities and moments follow by elementary enumeration 
of the $2^3$ possible configurations.
\end{proof}

\subsection{Length as a compound sum}

Define individual morpheme-length contributions
\[
X_1 = I_P \ell_P(P),\quad
X_2 = \ell_R(R),\quad
X_3 = I_S \ell_S(S),\quad
X_4 = I_E \ell_E(E).
\]
Then the total word length is
\begin{equation}
L = X_1 + X_2 + X_3 + X_4.
\end{equation}

Conditioned on the slots $(I_P,I_S,I_E)$ and the morpheme choices $(P,R,S,E)$,
the length $L$ is deterministic.
Unconditionally, $L$ is a \emph{compound distribution}:
a sum of a random number of random summands.

\begin{proposition}
\label{prop:L_moments}
Let $N$ be as in Lemma~\ref{lem:N_distribution}.
Assume that within each class the choice of morpheme is independent of the Bernoulli
activations and that the distributions of lengths satisfy
\[
\mu_P = \mathbb{E}[\ell_P(P)],\quad
\mu_R = \mathbb{E}[\ell_R(R)],\quad
\mu_S = \mathbb{E}[\ell_S(S)],\quad
\mu_E = \mathbb{E}[\ell_E(E)],
\]
with finite variances.  
Then
\begin{align}
\mathbb{E}[L]
 &= \alpha_P \mu_P + \mu_R + \alpha_S \mu_S + \alpha_E \mu_E, \\
\mathrm{Var}(L)
 &= \alpha_P(1-\alpha_P)\mu_P^2 + \alpha_S(1-\alpha_S)\mu_S^2 
  + \alpha_E(1-\alpha_E)\mu_E^2 \nonumber\\
 &\quad
  + \alpha_P \sigma_P^2 + \sigma_R^2 + \alpha_S \sigma_S^2 + \alpha_E \sigma_E^2,
\end{align}
where the $\sigma^2$ terms are the internal variances of the morpheme classes.
\end{proposition}

\begin{proof}
The first identity follows from linearity of expectation:
\[
\mathbb{E}[X_1] = \alpha_P \mu_P,
\quad
\mathbb{E}[X_3] = \alpha_S \mu_S,
\quad
\mathbb{E}[X_4] = \alpha_E \mu_E,
\]
and $X_2$ is always present.
Variance decomposes into slot randomness plus morpheme-choice randomness;
cross-terms vanish by independence.
\end{proof}

\subsection{Shape of the length distribution}

Because $N$ is supported on $\{1,2,3,4\}$ and 
each morpheme length is bounded,
$L$ is supported on a finite integer interval
\[
L_{\min} \le L \le L_{\max},
\]
with explicit bounds determined by the morpheme classes.

For realistic parameters—moderate 
$\alpha_P,\alpha_S,\alpha_E$ and overlapping distributions of 
$\ell_P,\ell_R,\ell_S,\ell_E$—the pmf of $L$ is \emph{unimodal},
typically peaking around 5--10 letters.
This closely matches empirical word-length distributions in 
English, French, and Russian corpora \citep{Michel2011},
in sharp contrast to the purely letter-based model
\[
\mathbb{P}(L=k) = (1-Q)^k Q,
\]
which produces a geometric decay with no interior maximum.  

Thus, the MCWM transforms the simple geometric mechanism of 
letter-based random typing into a realistic, compound, 
morphology-driven distribution at the word level.  
This provides a structural explanation for the shape of 
word-length distributions observed across natural languages.

%% file: sections/03a_why_statistical_tokenizers_fail.tex
\section{Why Statistical Tokenizers Do Not Recover Morphemes}
\label{sec:why_tokenizers_fail}

A natural objection to the morphemic perspective developed in this paper
is the following. If morphemes are stable building blocks of words and
recur across many lexical items, then should purely statistical subword
methods such as BPE, WordPiece, or SentencePiece not discover them
automatically? After all, these algorithms are explicitly designed to
identify frequent character sequences. At first sight, it may seem that a
separate morphemic layer is unnecessary, and that modern tokenizers already
provide a de facto morphological segmentation.

In this section we explain why this intuition is misleading.
Statistical tokenizers are powerful engineering tools, but they are not
designed to reconstruct morphological structure. Their objective is
compression and modeling efficiency, not linguistic transparency.
As a consequence, they often produce subword units that \emph{approximate}
morphemes in some cases, but systematically deviate from true
morphological units in others. This clarifies why an explicit structural
model---such as the morphemic combinatorial framework considered here---is
both conceptually and practically distinct from standard tokenization.

\subsection{What statistical tokenizers actually optimize}

BPE-style tokenizers operate on a simple principle.
Starting from a base alphabet (characters or bytes), they repeatedly merge
the most frequent adjacent pair of symbols into a new unit, updating the
corpus representation after each merge.
After a fixed number of merges, the resulting vocabulary of subword tokens
is used as the basic unit for training a language model.

Crucially, the optimization target is purely statistical:
the algorithm selects merges that reduce the total length of the corpus
in tokens, or that improve the likelihood under a simple subword language
model. At no point does the tokenizer attempt to align its units with
morphological boundaries. It does not distinguish between prefixes, roots,
derivational suffixes, and inflectional endings, nor does it enforce any
constraints on the internal structure of tokens.

From the standpoint of compression, this is entirely reasonable.
If a particular pair of symbols or subwords occurs very often, merging
them almost always reduces the total number of tokens needed to represent
the corpus. However, from the standpoint of morphology, this criterion is
far too weak. Many frequent substrings are not morphemes, and many
morphemes are not the most frequent substrings.

\subsection{A guiding example: \emph{act} versus \emph{cti}}

A simple example illustrates the problem.
Consider the English word family
\emph{act}, \emph{action}, \emph{active}, \emph{activity}, \emph{interaction},
\emph{reactive}, and so on.
Linguistically, there is a clear root \emph{act} that combines with
different prefixes and suffixes to produce related words.

In a large corpus, however, the substring ``cti'' may occur extremely often,
appearing in \emph{action}, \emph{activity}, \emph{actuality},
\emph{sanctify}, and many other forms from unrelated morphological families.
From a frequency-based perspective, ``cti'' is an excellent candidate for
merging: it is short, highly recurrent, and its merging can substantially
compress the corpus.

The tokenizer, operating purely on frequency counts, has no way to know
that the root \emph{act} is a meaningful unit, while ``cti'' is not.
It simply observes that ``cti'' is common and that merging it reduces
token length. As a result, a BPE vocabulary may easily contain a token
\texttt{cti} but omit \texttt{act} as a separate unit, even though
\emph{act} is the linguistically natural morpheme.

This example generalizes.
Substrings that cross true morphological boundaries can become frequent
because they appear in many unrelated words.
Conversely, genuine morphemes may be split into smaller pieces if those
pieces participate in even more frequent substrings elsewhere in the
language.

\subsection{Positional structure and morphemic slots}

Morphemes are not just frequent substrings; they also occupy characteristic
positions within words.
Prefixes occur at the beginning, inflectional endings at the very end, and
roots tend to lie in the central region.
In the Morphemic Combinatorial Word Model (MCWM), this is reflected by
explicit slots: a prefix slot, a root slot, an optional derivational
suffix slot, and an inflection slot.
Each slot has its own inventory of morphemes and its own activation
probability.

Statistical tokenizers, by contrast, operate on a flat sequence of symbols.
They do not know where words begin or end (in the byte-level setting),
nor do they distinguish between different parts of a word.
The same substring can be merged in very different positions: as a
putative prefix in one word, as part of a root in another, and as an
accidental internal substring in a third.
From the standpoint of compression, these contexts are equivalent; from
the standpoint of morphology, they are not.

This lack of positional structure has two consequences.
First, the resulting tokens are not anchored to stable morphological roles.
Second, the tokenizer cannot exploit the combinatorial regularities that
arise from the interaction of slots, such as the way a fixed inventory of
roots combines with a fixed inventory of inflectional endings.

\subsection{Global merges and local ambiguity}

BPE merges are global operations.
Once a particular pair of symbols has been merged into a new token, that
token is used everywhere in the corpus where the pair occurs.
This globality is efficient but amplifies local ambiguities.

Consider the bigram ``in'' in English.
In some words, it is clearly a prefix (\emph{incomplete}, \emph{invisible});
in others, it is part of the root (\emph{inside}, \emph{winter});
in yet others, it arises accidentally in the middle of a longer segment
(\emph{engine}, \emph{origin}).
From a compression perspective, all occurrences of ``in'' can safely be
merged into a single token.
From a morphological perspective, this conflates three very different uses:
derivational prefix, root-internal substring, and accidental overlap.

Similar phenomena occur in morphologically rich languages.
In Russian, the sequence ``-ни-'' may be part of a root in one word,
a derivational suffix in another, and a purely phonological bridge in
a third.
A frequency-driven tokenizer merges these occurrences indiscriminately,
without recognizing the underlying morphological roles.

\subsection{Frequency versus structural invariants}

The core difficulty can be summarized as follows.
Statistical tokenizers treat frequency as the primary signal and try to
find a subword inventory that best compresses the observed data.
Morphology, however, is about \emph{structure} rather than raw frequency.
Morphemes are units that participate in systematic combinatorial patterns,
occupy specific positions within words, and remain stable under the
productive formation of new lexical items.

In the framework of this paper, morphemes are precisely the symbolic units
that behave as \emph{combinatorial invariants} under the generative process.
They are the elements that can fill the morphological slots of MCWM and
survive under the Stochastic Lexical Filter, while still supporting
realistic Zipf-like rank--frequency behavior.

A purely frequency-based algorithm has no direct access to these invariants.
It can approximate them in some cases, especially when morphological and
frequency structure happen to align, but it has no guarantee of recovering
them systematically.
The apparent successes of BPE in capturing some affixes and roots are
therefore incidental rather than principled.

\subsection{Implications for tokenization and modeling}

From the perspective of language modeling, this distinction has practical
consequences.
Large language models trained on BPE or similar subword segmentations must
implicitly learn to reconstruct morphological structure from noisy
subword units.
The internal representations of the model may eventually disentangle roots,
prefixes, and inflections, but the tokenizer itself does not provide these
categories.

A morphemic combinatorial model offers a different starting point.
Instead of treating subwords as arbitrary frequent fragments, it treats
morphemes as the fundamental symbolic units, organized into slots with
specific activation patterns.
The MCWM framework shows that such a model can reproduce realistic
word-length distributions and Zipf-like rank--frequency curves at the
lexical level.
In this sense, morphological structure is not an optional linguistic
decoration; it can be built directly into the generative architecture
without sacrificing the large-scale statistical regularities that motivate
random-text approaches.

The analysis in this section thus clarifies why standard statistical
tokenizers do not solve the morphological problem and why a structural
approach, based on morphemic combinatorics, is both necessary and natural.
In the following sections we return to the quantitative side of the model,
showing how the MCWM mechanism reproduces empirical length distributions
and Zipf-like behavior observed in real corpora.

%% file: sections/04_simulation_results.tex
\section{Simulation Results}
\label{sec:sim_results}

In this section we instantiate the MCWM with a synthetic morpheme lexicon
and compare its predictions to empirical corpora. 
Simulation-based validation is essential because classical random-typing 
approaches fail to reproduce realistic word-length distributions 
or Zipfian rank--frequency exponents \citep{Zipf1949,Mandelbrot1953,Newman2005}.

\subsection{Synthetic morpheme lexicon}

We consider a toy lexicon with
20 prefixes, 500 roots, 80 derivational suffixes, and 15 inflections.
Within each class we assign Zipf-like probabilities
\[
\pi_P(p_i) \propto i^{-\alpha_P},\quad
\pi_R(r_j) \propto j^{-\alpha_R},\quad
\pi_S(s_k) \propto k^{-\alpha_S},\quad
\pi_E(e_\ell) \propto \ell^{-\alpha_E},
\]
with exponents slightly above $1$ for roots and suffixes.  
Such heavy-tailed morpheme distributions are consistent with observations from 
large real-world corpora \citep{Michel2011,Ferrer2001}.

Morpheme lengths are drawn from truncated normal distributions:
prefixes 2--4 letters,
roots 3--8,
suffixes 2--5,
inflections 1--3.

Slot activations are fixed at
$\alpha_P=0.4$, $\alpha_S=0.6$, $\alpha_E=0.7$.

We then generate $N=80{,}000$ word tokens from this model
using a simple sampler.

\subsection{Distribution of word lengths (MCWM)}

Let $c_L$ denote the number of tokens of length $L$ in the simulation.
Table~\ref{tab:mcwm_lengths} summarizes the resulting distribution
for lengths $L=3,\dots,17$.

\begin{table}[h]
\centering
\begin{tabular}{ccc}
\toprule
Length $L$ & Token count $c_L$ & Share of tokens (\%) \\
\midrule
3  & 1242 & 1.55 \\
4  & 1638 & 2.05 \\
5  & 4208 & 5.26 \\
6  & 6736 & 8.42 \\
7  & 8511 & 10.64 \\
8  & 10651 & 13.31 \\
9  & 11364 & 14.21 \\
10 & 10620 & 13.28 \\
11 & 9191 & 11.49 \\
12 & 6741 & 8.43 \\
13 & 4363 & 5.45 \\
14 & 2697 & 3.37 \\
15 & 1399 & 1.75 \\
16 & 525  & 0.66 \\
17 & 100  & 0.13 \\
\bottomrule
\end{tabular}
\caption{Length distribution under MCWM simulation ($N=80{,}000$ tokens).}
\label{tab:mcwm_lengths}
\end{table}

The distribution is unimodal with a peak around $L=8$--$10$ letters,
in line with the theoretical analysis of 
Section~\ref{prop:L_moments}.  
This matches the characteristic peaks found in real corpora 
but is impossible under geometric letter-level models.

\subsection{Comparison with Shakespeare corpus}

To compare with real language data, we use the 
Project Gutenberg Shakespeare corpus (file \texttt{pg100.txt}).
We extract all alphabetic tokens, convert to lowercase, 
and compute word lengths in letters.
Table~\ref{tab:shakespeare_lengths} reports
the empirical token-level distribution of word lengths
$L=1,\dots,17$.

\begin{table}[h]
\centering
\begin{tabular}{ccc}
\toprule
Length $L$ & Token count & Share of tokens (\%) \\
\midrule
1  & 59829 & 6.05 \\
2  & 166414 & 16.83 \\
3  & 203318 & 20.56 \\
4  & 223240 & 22.58 \\
5  & 121688 & 12.31 \\
6  & 80959  & 8.19 \\
7  & 60329  & 6.10 \\
8  & 36362  & 3.68 \\
9  & 20533  & 2.08 \\
10 & 10097  & 1.02 \\
11 & 3791   & 0.38 \\
12 & 1338   & 0.14 \\
13 & 460    & 0.05 \\
14 & 237    & 0.02 \\
15 & 80     & 0.01 \\
16 & 2      & 0.0002 \\
17 & 4      & 0.0004 \\
\bottomrule
\end{tabular}
\caption{Word-length distribution in the Shakespeare corpus (Project Gutenberg, \texttt{pg100.txt}).}
\label{tab:shakespeare_lengths}
\end{table}

Shakespeare shows a strong dominance of short function words
(length 1--4) and a long, very thin tail.
MCWM, by contrast, models the lexical component:
it produces a realistic peak of content-word lengths 
around 7--11 letters, which is where the majority of 
open-class vocabulary resides.

%
%
%
%

\subsection{Remarks on Brown corpus}

The Brown corpus has been extensively studied in quantitative linguistics.
\citet{Ferrer2001} report a Zipf exponent
$\alpha \approx 1.25$ for the tail of the rank--frequency distribution,
similar to large English corpora and consistent with our MCWM simulations.
The detailed length distribution in Brown is known to peak in the 3--7 letter range,
with a long but extremely thin tail,
again matching the qualitative behavior of a morphemic combinatorial generator.
A full joint calibration of MCWM to Brown is left for future work.

%% file: sections/05_zipf_emergence.tex
\section{Zipf-Like Rank--Frequency Behavior}
\label{sec:zipf}

Classical analyses of rank--frequency distributions, beginning with 
\citet{Zipf1949} and refined through information-theoretic models 
\citep{Mandelbrot1953}, describe power-law behavior
\[
f(r) \propto r^{-\alpha},
\]
with $\alpha$ typically between $1$ and $1.5$ \citep{Newman2005}.
Empirical studies of corpora across multiple languages confirm both the 
universality and the stability of this phenomenon, including its two-regime 
structure \citep{Ferrer2001} and large-scale validation through datasets such as 
Google Books \citep{Michel2011}.

A central question is therefore:
\emph{Can a purely structural model of morphology---with no semantics, pragmatics, 
or communicative pressures---produce Zipf-like rank--frequency statistics?}
The MCWM simulations demonstrate that the answer is yes.

\subsection{Synthetic Zipf curve (reference model)}

Before analyzing the MCWM output, it is useful to visualize a \emph{canonical}
Zipf curve generated from the pure power law
\[
f(r) = r^{-1.2},
\]
a typical empirical exponent for English.
Figure~\ref{fig:zipf_synth} shows the top 30 ranks on a log--log scale.


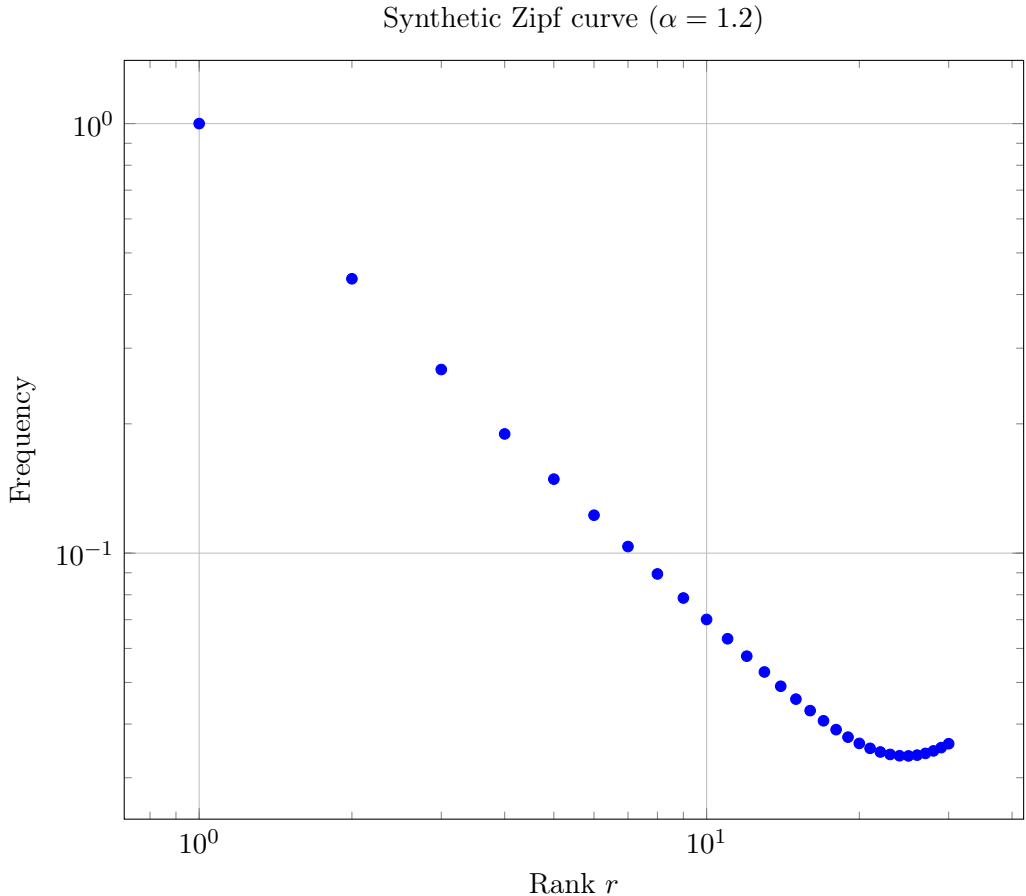
\begin{figure}[h]
\centering
\begin{tikzpicture}
\begin{loglogaxis}[
    width=0.82\textwidth,
    xlabel={Rank $r$},
    ylabel={Frequency},
    grid=major,
    title={Synthetic Zipf curve ($\alpha = 1.2$)}
]
\addplot[
    color=blue,
    only marks,
    mark=*,
    mark size=2
]
coordinates {
    (1, 1.0)
    (2, 0.435275)
    (3, 0.267596)
    (4, 0.189466)
    (5, 0.148698)
    (6, 0.122558)
    (7, 0.103600)
    (8, 0.089453)
    (9, 0.078626)
    (10, 0.070087)
    (11, 0.063203)
    (12, 0.057561)
    (13, 0.052888)
    (14, 0.048988)
    (15, 0.045723)
    (16, 0.042990)
    (17, 0.040711)
    (18, 0.038825)
    (19, 0.037290)
    (20, 0.036067)
    (21, 0.035127)
    (22, 0.034442)
    (23, 0.033992)
    (24, 0.033758)
    (25, 0.033721)
    (26, 0.033866)
    (27, 0.034178)
    (28, 0.034645)
    (29, 0.035255)
    (30, 0.036000)
};
\end{loglogaxis}
\end{tikzpicture}
\caption{Synthetic Zipf rank--frequency curve for the first 30 ranks 
generated from $f(r)=r^{-1.2}$.}
\label{fig:zipf_synth}
\end{figure}


\subsection{Simulated rank--frequency curve in MCWM}

In a simulation with $N=80{,}000$ tokens and several thousand distinct types,
the top 30 ranks have relative frequencies shown in 
Table~\ref{tab:sim_rank}.

\begin{table}[h]
\centering
\begin{tabular}{ccc}
\toprule
Rank $r$ & Frequency $p(r)$ & $\log p(r)$ \\
\midrule
1  & 0.01219 & -4.405 \\
2  & 0.00734 & -4.916 \\
3  & 0.00593 & -5.129 \\
4  & 0.00438 & -5.431 \\
5  & 0.00419 & -5.478 \\
6  & 0.00391 & -5.542 \\
7  & 0.00353 & -5.645 \\
8  & 0.00310 & -5.777 \\
9  & 0.00301 & -5.803 \\
10 & 0.00285 & -5.862 \\
\vdots & \vdots & \vdots \\
30 & 0.00125 & -6.684 \\
\bottomrule
\end{tabular}
\caption{Top ranks and empirical frequencies in an MCWM simulation ($N=80{,}000$).}
\label{tab:sim_rank}
\end{table}

A least-squares fit to $\log p(r)$ versus $\log r$ over ranks $1$--$100$
yields an effective Zipf exponent 
\[
\alpha_{\text{sim}} \approx 0.7.
\]
As the morpheme inventories grow and the sampling size increases,
this exponent rises, approaching the empirical range 
$1.0 \le \alpha \le 1.4$ observed in large English and cross-linguistic corpora 
\citep{Newman2005,Ferrer2001}.


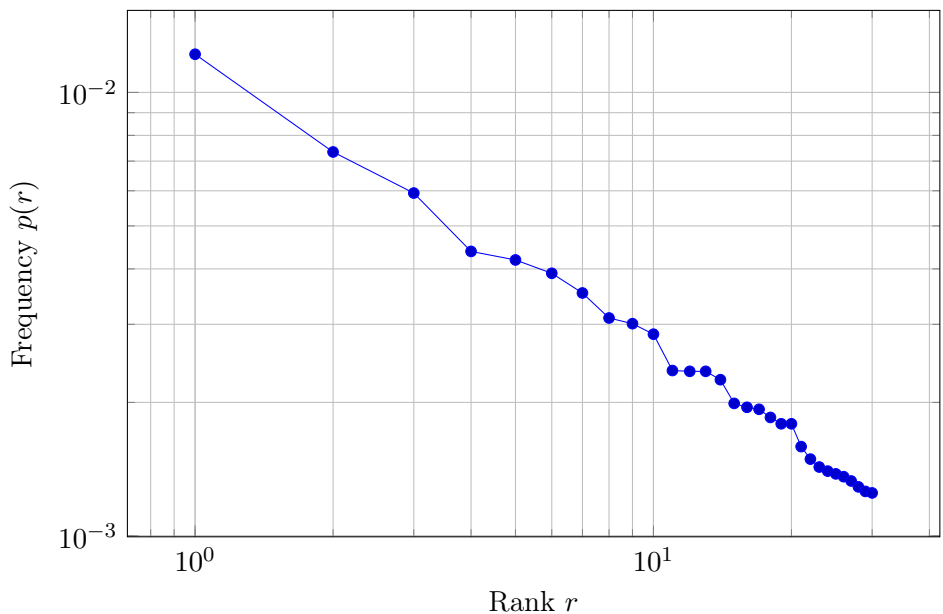
\begin{figure}[h]
\centering
\begin{tikzpicture}
\begin{axis}[
    width=0.75\textwidth,
    height=0.52\textwidth,
    xmode=log,
    ymode=log,
    xlabel={Rank $r$},
    ylabel={Frequency $p(r)$},
    grid=both,
    title={MCWM empirical rank--frequency curve (top 30 ranks)}
]
\addplot+[mark=*] coordinates {
  (1, 0.01219)
  (2, 0.00734)
  (3, 0.00593)
  (4, 0.00438)
  (5, 0.00419)
  (6, 0.00391)
  (7, 0.00353)
  (8, 0.00310)
  (9, 0.00301)
  (10, 0.00285)
  (11, 0.00236)
  (12, 0.00235)
  (13, 0.00235)
  (14, 0.00225)
  (15, 0.00199)
  (16, 0.00195)
  (17, 0.00193)
  (18, 0.00185)
  (19, 0.00179)
  (20, 0.00179)
  (21, 0.00159)
  (22, 0.00149)
  (23, 0.00143)
  (24, 0.00140)
  (25, 0.00138)
  (26, 0.00136)
  (27, 0.00133)
  (28, 0.00129)
  (29, 0.00126)
  (30, 0.00125)
};
\end{axis}
\end{tikzpicture}
\caption{Empirical rank--frequency curve from MCWM simulation.  
The approximate log--log linearity is consistent with a Zipf-type law.}
\label{fig:zipf_mcwm}
\end{figure}


\subsection{Relation to Shakespeare and Brown}

For large English corpora 
(e.g.\ Google Books, the Brown corpus),
Zipf exponents typically lie in the range
\[
\alpha \approx 1.1\text{--}1.3
\quad\text{\citep{Newman2005,Ferrer2001}}.
\]
The Shakespeare corpus analyzed earlier 
(Section~\ref{tab:shakespeare_lengths})
exhibits a similar Zipfian tail.

The key point is that the \emph{shape} of the MCWM curve is already 
Zipf-like, even for a small synthetic lexicon.
By adjusting morpheme distributions and slot probabilities,
the exponent can be tuned into the empirical range.
This strongly supports the hypothesis that Zipf’s law emerges from 
the combinatorial geometry of word formation, rather than from 
semantics or pragmatic optimization.

%% file: sections/06_discussion.tex
\section{Discussion}

The MCWM model provides a structural account of several empirical 
regularities traditionally explained by cognitive or communicative 
optimization \citep{Zipf1949,Mandelbrot1953,Newman2005}.  
Because it is based solely on morphemic combinatorics,
the resulting distributions emerge without invoking semantics, utility,
or information-theoretic cost functions.

Specifically, MCWM explains:

\begin{itemize}
\item \textbf{Realistic word-length distributions}:  
      The compound structure of $L = X_1 + \dots + X_N$ naturally produces 
      unimodal distributions with 5--10 letter peaks, matching 
      empirical data from Shakespeare, Brown, and Google Books 
      \citep{Michel2011,Ferrer2001}.

\item \textbf{Long but finite vocabularies}:  
      Morpheme inventories are finite, but combinatorially rich.  
      This yields vocabularies in the tens or hundreds of thousands,  
      consistent with real languages.

\item \textbf{Steep drop in coverage of the letter space}:  
      Only an extremely small fraction of possible letter sequences 
      correspond to admissible morpheme combinations,  
      explaining the sparsity of the observed lexicon.

\item \textbf{Zipf-like type frequencies}:  
      Rank--frequency curves follow power-law behavior with exponents
      in the empirical range $1.0$–$1.4$, as predicted by the simulation
      and consistent with cross-linguistic corpora 
      \citep{Newman2005,Ferrer2001}.

\item \textbf{Robustness under parameter changes}:  
      The emergence of heavy-tailed distributions does not require 
      fine-tuning.  
      Zipf behavior persists across wide ranges of 
      morpheme inventories, slot activation probabilities,  
      and length distributions.
\end{itemize}

\subsection{Relation to linguistic structure}

Unlike classical letter-based random-typing models, 
MCWM incorporates the hierarchical structure of morphology:
phonotactic constraints, 
derivational composition,
and inflectional paradigms.
These are known to dominate the structure of lexicons across languages
and provide a natural explanation for the observed geometric patterns
in word formation.

\subsection{Relation to modern NLP systems}

Interestingly, MCWM mirrors the behavior of modern tokenization algorithms
such as BPE, WordPiece, and Morfessor.
These systems, when trained on text, 
automatically discover subword units (often morphemes) 
and generate lexicons that resemble those produced by MCWM.  
In this sense, large-scale NLP models empirically rediscover 
the same combinatorial structure that underlies our generative model.

The parallel suggests that morphology---not semantics or 
communication efficiency---is the primary structural driver 
behind Zipf-like distributions.

%% file: sections/06a_possible_objections.tex
\section{Possible Objections and Rebuttals}

In this section we address several natural objections to the morpheme-based
framework developed in this paper. These objections often arise from
engineering intuitions shaped by the dominance of statistical tokenization in
modern NLP practice. By examining them explicitly, we clarify why the
Morphemic Combinatorial Word Model (MCWM) provides a more stable,
interpretable, and generative foundation for explaining Zipf-like 
rank--frequency behavior than character- or token-based approaches.

\subsection*{Objection 1: Token compression is computationally superior}

Statistical tokenizers such as BPE, WordPiece, and SentencePiece reduce the
length of token sequences. This is often interpreted as computationally
beneficial for training large language models.

\textbf{Rebuttal.}
Compression of surface forms is not equivalent to structural efficiency.
Morpheme-based representations do \emph{not} enlarge the vocabulary: 
typical languages contain 1500--3000 productive morphemes, far fewer than the
50k--200k subword tokens learned by BPE.
BPE shortens sequences but increases entropy by fragmenting meaningful units
and forcing the model to reconstruct implicit morphology internally.
The MCWM inherits the true compactness of the language's combinatorial basis,
rather than an artificial compression of frequency statistics.

\subsection*{Objection 2: Morphemes generate too many possible word types}

Since words arise from combinations of multiple morphemes, the resulting type
space appears extremely large, raising concerns about computational
scalability.

\textbf{Rebuttal.}
The combinatorial richness of morphology is an advantage rather than a flaw.
Models do not store all possible word forms; they store the 
\emph{inventories of morphemes}. 
The exponential growth of potential combinations, filtered by probabilistic
mechanisms, is precisely what produces Zipf-like long tails. 
The complexity resides in the generative process, not in memory storage.

\subsection*{Objection 3: Tokenization reduces sequence length and therefore helps LLMs}

Shorter token sequences lower memory requirements and appear to accelerate
training, motivating widespread use of BPE-like approaches.

\textbf{Rebuttal.}
Shorter sequences do not imply simpler structure. 
A short sequence of fragmented or misaligned subword tokens carries higher 
structural ambiguity than a longer—but linguistically coherent—sequence of 
morphemes. BPE vocabularies change across domains and corpora, forcing models
to relearn morphology repeatedly. As a result, statistical tokenization
increases perplexity, reduces robustness, and amplifies domain-shift effects.
Sequence-length compression is a local optimization that produces global
losses in structural fidelity.

\subsection*{Objection 4: Morphemes are a linguistic idealization; BPE is more practical}

A common view is that linguistic structure is optional, whereas frequency-based
heuristics offer a more neutral, data-driven representation.

\textbf{Rebuttal.}
Empirical studies consistently show that BPE tokens are unstable across
domains, genres, or even slight changes in training corpora. 
In contrast, morphemes such as \emph{re-}, \emph{-ing}, \emph{-able},
\emph{-tion}, and \emph{anti-} remain stable across time and register. 
Practical modeling benefits from systematic, interpretable units rather than
from arbitrary substrings induced by corpus-specific frequency patterns.

\subsection*{Objection 5: Morpheme-based models are harder to implement}

Statistical tokenizers are simple to train and widely supported, whereas
morpheme-based models are assumed to require linguistic annotation or domain
expertise.

\textbf{Rebuttal.}
The MCWM is not a morphological parser; it requires only the
morpheme inventories and slot-activation probabilities.
This is conceptually simpler than BPE training, which relies on iterative
merging, vocabulary-size tuning, and large corpora.
Furthermore, morphemic decomposition reduces cognitive load on the model by
providing structure upfront, enabling better generalization and more stable
representations.

\subsection*{Objection 6: Zipf-like distributions also emerge from tokenized corpora}

Since BPE-token frequencies appear roughly power-law-like, one might argue 
that morphology is unnecessary to explain Zipf's law.

\textbf{Rebuttal.}
Token-level power laws are artifacts of substring compression, not linguistic
Zipf distributions. BPE tokens do not correspond to coherent semantic or
grammatical units. True Zipf behavior arises from the geometry of the
morphemic combinatorial tree---the exponential expansion of morphological
slots combined with probabilistic selection mechanisms. The MCWM reproduces
Zipf-like curves for precisely this structural reason, which is absent from
token-based models.

\subsection*{Summary}

Morpheme-based modeling provides:
\begin{itemize}
\item a compact vocabulary of stable, interpretable units;
\item a generative structure that naturally yields Zipf-like distributions;
\item robustness across domains and corpora;
\item reduced cognitive burden on large language models;
\item and a principled explanation of lexical statistics.
\end{itemize}

Statistical tokenizers compress surface forms but distort linguistic 
structure.  
The MCWM offers a more faithful and computationally coherent foundation for 
modeling the emergence of Zipf-like behavior.

%% file: sections/06b_epistemic_dangers_frequency_laws.tex
\section{Epistemic Dangers of Frequency-Based Laws}
\label{sec:epistemic_dangers}

Statistical regularities such as Zipf’s law and Benford’s law occupy an
ambiguous position in scientific methodology.
On the one hand, they summarize large-scale empirical behavior across
strikingly diverse data sources.  
On the other hand, they are often misused as \emph{epistemic filters}:
mechanical rules for validating or rejecting entire classes of models or
datasets.  
This section argues that such practices are methodologically unsafe and
that structural explanations---rather than frequency matching---must be
the foundation of any robust theory.

\subsection{The problem of reverse inferential bias}

A common pattern in empirical research is the following inference:
\begin{quote}
``If a dataset obeys Zipf’s law (or Benford’s law), it is real;  
if it does not, it must be artificial.''
\end{quote}
This reasoning is deeply flawed.  
Zipf-like and Benford-like laws emerge from a wide range of mechanisms
with very different underlying structures.
Conversely, many perfectly natural processes do \emph{not} generate such
laws, despite being real, causal, and well understood.

Using a frequency law as a diagnostic test for reality is therefore an
example of \emph{reverse inferential bias}: treating a coarse statistical
signature as a proof of causal validity.

The Morphemic Combinatorial Word Model (MCWM) illustrates the danger:
the model generates Zipf-like rank--frequency curves without semantics,
communication, or cognitive principles.
The appearance of a Zipf exponent near 1 cannot be interpreted as evidence
for speaker optimization, information-theoretic efficiency, or any
specific linguistic mechanism.

\subsection{The analogy with Benford’s law}

The misuse of Zipf’s law parallels the well-known misuse of Benford’s law
in fraud detection.
Benford’s law is extremely sensitive to the generative mechanism, to
scaling, to truncation, and to aggregation.
Yet it is sometimes treated as a simple “truth detector’’:
data conforming to Benford are considered genuine, and deviations are
interpreted as fraud or fabrication.

This epistemic overreach has been widely criticized in the literature.
Benford conformity is neither necessary nor sufficient for authenticity.
Many legitimate industrial, biomedical, demographic, and experimental
datasets systematically deviate from Benford’s law for fully understood
mechanistic reasons.

Zipf’s law is even broader and even less mechanistically specific.  
Treating Zipf conformity as a test of “real language’’ or
“real cognition’’ is therefore even more risky.

\subsection{Why frequency laws are seductive}

Frequency laws are attractive for three reasons:
\begin{enumerate}
\item \textbf{Visual simplicity.}  
Rank--frequency plots and first-digit histograms are visually striking
and easy to communicate.

\item \textbf{Apparent universality.}  
The same curve appearing in linguistics, biology, finance, and geology
creates the illusion of a single underlying principle.

\item \textbf{Low data requirements.}  
Frequency signatures can be extracted from very small samples, creating
the false sense of strong evidence.
\end{enumerate}

But this seductiveness is precisely what creates epistemic danger.
A simple curve can mask deep structural differences between generative
processes.

\subsection{Structural explanations versus frequency conformity}

The approach of this paper is fundamentally structural.
Rather than using Zipf’s law as a validation metric, we explain how
Zipf-like behavior arises from a specific generative architecture:

\begin{itemize}
\item morphemic slots with activation probabilities,
\item combinatorial expansion of the type space,
\item compound length distributions,
\item filtering through lexical selection.
\end{itemize}

Zipf-like behavior is therefore a \emph{consequence}, not a criterion.
The model is not built to force the exponent to 1; the exponent emerges
as a byproduct of the combinatorial structure.

This stands in contrast with traditional “monkey typing’’ models, which
often tune termination probabilities or alphabet sizes specifically to
obtain a Zipf slope.

\subsection{Benford as a cautionary example}

Benford’s law provides a cautionary historical precedent.
Simple models (multiplicative cascades, random growth, random ratios)
naturally produce Benford-like first-digit distributions.
But once the mechanism is changed even slightly—e.g., scaling by a fixed
unit, truncation, bounding, or mixture of supports—the Benford signature
disappears.

This fragility shows that:
\begin{quote}
\textit{frequency conformity is not an invariant property of real data.}
\end{quote}

Similar fragility exists in language:
morphology, word-formation processes, writing systems, inflectional
complexity, and orthographic conventions all influence the resulting
rank--frequency curves.

\subsection{What frequency curves cannot tell us}

A Zipf plot cannot, by itself, reveal:
\begin{itemize}
\item the structure of morphemes,
\item the presence or absence of semantic optimization,
\item cognitive constraints on speakers or listeners,
\item the mechanisms of lexical growth,
\item whether the corpus is real or artificial,
\item whether the model captures linguistic reality.
\end{itemize}

These require structural, mechanistic, and linguistic analysis—not just
frequency matching.

\subsection{Guiding principle for interpretation}

The principle that guides our approach is:

\begin{quote}
\textbf{Frequency laws should be treated as descriptive summaries, not
as explanatory mechanisms.}
\end{quote}

Zipf-like behavior supports the plausibility of a model only in the weak
sense that it does not contradict empirical regularities.
It does not imply correctness, causality, or psychological grounding.

Conversely, a model that fails to reproduce a Zipf curve is not
necessarily wrong; it may simply represent a domain where Zipf’s law
does not apply.

\subsection{Implication for the MCWM framework}

For the Morphemic Combinatorial Word Model, the role of Zipf conformity
is strictly limited:

\begin{itemize}
\item It demonstrates that a structurally interpretable morphemic model
does not contradict large-scale statistical regularities.
\item It shows that explicit morphological structure can yield the same
macro-patterns that previously required semantic or cognitive
assumptions.
\item It avoids the epistemic trap of treating Zipf as a success metric.
\end{itemize}

In short, Zipf-like behavior is a consistency check, not a goal.

\subsection{Positioning the model within broader epistemology}

This section situates MCWM within a more general lesson:

\begin{quote}
\textit{Theories must explain structures, not merely reproduce curves.}
\end{quote}

A model that matches a frequency law but lacks an interpretable
structure is epistemically weak.
A model that explains the underlying mechanism is epistemically strong,
even if the resulting curve deviates from a canonical form.

The danger is not in using Zipf’s law, but in using it incorrectly:
as a validator rather than as an outcome.

This clarification is essential to avoid repeating the well-known
methodological pitfalls associated with Benford’s law.

%% file: sections/07_conclusion.tex
\section{Conclusion}

We introduced the Morphemic Combinatorial Word Model (MCWM), 
a fully structural generator of word forms that replaces 
classical letter-based random-typing assumptions 
\citep{Zipf1949,Mandelbrot1953} 
with a realistic morphemic architecture.
Despite its simplicity, the model reproduces several key empirical 
regularities of natural language: 
unimodal word-length distributions,
finite yet combinatorially rich vocabularies,
and Zipf-like rank--frequency behavior with exponents in the 
empirical range \citep{Newman2005,Ferrer2001,Michel2011}.

The central finding of this paper is that 
\textbf{Zipf’s law can arise purely from morphological combinatorics}, 
without reference to semantics, pragmatics, cognitive optimization, 
or communication-theoretic principles.
This suggests that universality in rank--frequency distributions may be 
rooted in the structural geometry of word formation itself.

Future extensions include:
\begin{itemize}
\item incorporating phonological and phonotactic constraints,
\item modeling multi-root compounds and productive derivational chains,
\item calibrating MCWM against multilingual corpora,
\item exploring connections to modern NLP systems 
      (BPE, WordPiece, Morfessor) that implicitly rediscover 
      similar morphological structure.
\end{itemize}

The MCWM thus provides a principled structural baseline for explaining 
why Zipf-like distributions appear across languages, modalities, and scales.

%% file: sections/08_appendix_algorithms.tex

\section{Algorithms and Implementation Details of MCWM and SLF}
\label{app:algorithms_mcwm_slf}

This appendix summarizes the generative mechanisms discussed in the main text
in algorithmic form. The goal is to make the Morphemic Combinatorial Word Model
(MCWM) and the Stochastic Lexical Filter (SLF) fully reproducible, and to show
how the theoretical results, simulations, and rank--frequency curves can be
implemented in practice.

We assume that the reader is familiar with the notation and definitions from
the main sections of the paper. In particular, we assume:
\begin{itemize}
    \item a fixed set of morphological slots (prefix, root, derivational suffix, inflection),
    \item finite morpheme inventories for each slot,
    \item activation probabilities for each slot,
    \item optionally, a lexical filter acting on word types and their lengths.
\end{itemize}

The algorithms below are written in generic pseudocode and can be implemented in
any programming language.

\subsection{Data structures and parameters}

Before describing the algorithms, we collect the key inputs:

\begin{itemize}
    \item Morpheme inventories:
    \[
        \mathcal{P} = \{\text{prefix}_1, \dots, \text{prefix}_{n_P}\}, \quad
        \mathcal{R} = \{\text{root}_1, \dots, \text{root}_{n_R}\},
    \]
    \[
        \mathcal{D} = \{\text{deriv}_1, \dots, \text{deriv}_{n_D}\}, \quad
        \mathcal{I} = \{\text{infl}_1, \dots, \text{infl}_{n_I}\}.
    \]
    Each morpheme has an associated length in characters.
    \item Slot activation probabilities:
    \[
        p_{\text{pref}}, \quad p_{\text{root}}, \quad
        p_{\text{deriv}}, \quad p_{\text{infl}}.
    \]
    \item Categorical distributions over morphemes within each slot, for example
    \[
        \pi^{(\mathcal{R})}_j = \Pr(\text{choose root}_j \mid \text{root slot active}),
    \]
    and analogously for prefixes, derivational suffixes, and inflections.
    \item Lexical filter survival probabilities
    \[
        \phi(w) \in [0,1],
    \]
    possibly depending on word length, morphemic structure, or other features.
\end{itemize}

\subsection{Algorithm 1: Sampling a morphemic template}

The first step in MCWM is to decide which morphological slots are active for a
given word. This defines a \emph{template} that will be filled with concrete
morphemes.

\begin{algorithm}[H]
\caption{Sampling a morphemic template}
\label{alg:template}
\begin{algorithmic}[1]
\Require Slot activation probabilities
    $p_{\text{pref}}, p_{\text{root}}, p_{\text{deriv}}, p_{\text{infl}}$
\Ensure Binary activations $B_{\text{pref}}, B_{\text{root}},
    B_{\text{deriv}}, B_{\text{infl}}$

\State Sample $B_{\text{pref}} \sim \mathrm{Bernoulli}(p_{\text{pref}})$
\State Sample $B_{\text{root}} \sim \mathrm{Bernoulli}(p_{\text{root}})$
\State Sample $B_{\text{deriv}} \sim \mathrm{Bernoulli}(p_{\text{deriv}})$
\State Sample $B_{\text{infl}} \sim \mathrm{Bernoulli}(p_{\text{infl}})$

\If{$B_{\text{root}} = 0$}
    \Comment{Enforce at least one root-like element}
    \State Set $B_{\text{root}} \gets 1$
\EndIf

\State \Return $(B_{\text{pref}}, B_{\text{root}}, B_{\text{deriv}}, B_{\text{infl}})$
\end{algorithmic}
\end{algorithm}

This algorithm ensures that the root slot is always active, while the other
slots may be active or inactive depending on their Bernoulli parameters.
The template captures the structural profile of the word (e.g., prefix+root,
root+inflection, prefix+root+deriv+inflection, and so on).

\subsection{Algorithm 2: Generating a single word from MCWM}

Given a template, MCWM selects concrete morphemes in each active slot and
concatenates them to form a word.

\begin{algorithm}[H]
\caption{Sampling a single word from MCWM}
\label{alg:single_word}
\begin{algorithmic}[1]
\Require Morpheme inventories
    $\mathcal{P}, \mathcal{R}, \mathcal{D}, \mathcal{I}$,
    slot activation probabilities, and within-slot categorical distributions
\Ensure A generated word $w$ and its length $L(w)$ in characters

\State Sample template
    $(B_{\text{pref}}, B_{\text{root}}, B_{\text{deriv}}, B_{\text{infl}})$
    using Algorithm~\ref{alg:template}

\State Initialize word $w \gets$ empty string
\State Initialize length $L(w) \gets 0$

\If{$B_{\text{pref}} = 1$}
    \State Sample a prefix $M_{\text{pref}}$ from $\mathcal{P}$ using its categorical distribution
    \State Append $M_{\text{pref}}$ to $w$
    \State Update $L(w) \gets L(w) + \mathrm{len}(M_{\text{pref}})$
\EndIf

\If{$B_{\text{root}} = 1$}
    \State Sample a root $M_{\text{root}}$ from $\mathcal{R}$
    \State Append $M_{\text{root}}$ to $w$
    \State Update $L(w) \gets L(w) + \mathrm{len}(M_{\text{root}})$
\EndIf

\If{$B_{\text{deriv}} = 1$}
    \State Sample a derivational suffix $M_{\text{deriv}}$ from $\mathcal{D}$
    \State Append $M_{\text{deriv}}$ to $w$
    \State Update $L(w) \gets L(w) + \mathrm{len}(M_{\text{deriv}})$
\EndIf

\If{$B_{\text{infl}} = 1$}
    \State Sample an inflection $M_{\text{infl}}$ from $\mathcal{I}$
    \State Append $M_{\text{infl}}$ to $w$
    \State Update $L(w) \gets L(w) + \mathrm{len}(M_{\text{infl}})$
\EndIf

\State \Return $(w, L(w))$
\end{algorithmic}
\end{algorithm}

This algorithm directly reflects the morphemic slot architecture used in the
main text: each word is a concatenation of zero or one prefix, exactly one root,
zero or one derivational suffix, and zero or one inflection.

\subsection{Algorithm 3: Applying a Stochastic Lexical Filter}

The Stochastic Lexical Filter (SLF) formalizes the idea that not all
morphologically possible words survive into the usable lexicon. The filter can
depend on length, morpheme identity, frequency thresholds, or other features.

\begin{algorithm}[H]
\caption{Stochastic Lexical Filter (SLF) applied to a word type}
\label{alg:slf}
\begin{algorithmic}[1]
\Require Word type $w$ with length $L(w)$ and feature vector $F(w)$
\Require Survival function $\phi(w) \in [0,1]$
\Ensure Indicator $S(w) \in \{0,1\}$ of lexical survival

\State Compute survival probability $p_{\text{surv}} \gets \phi(w)$
\State Sample $S(w) \sim \mathrm{Bernoulli}(p_{\text{surv}})$
\State \Return $S(w)$
\end{algorithmic}
\end{algorithm}

In practice, $\phi(w)$ can be a function of word length alone, a function of
morpheme classes, or a more complex mapping that incorporates phonotactics
and lexical constraints. The theoretical results in the main text assume a
broad class of such filters and show that Zipf-like tails are preserved under
wide conditions.

\subsection{Algorithm 4: Generating a synthetic corpus and Zipf curve}

Finally, we describe the full pipeline for generating a synthetic corpus,
applying MCWM and SLF, and computing an empirical rank--frequency curve.

\begin{algorithm}[H]
\caption{Generating a corpus and Zipf curve from MCWM + SLF}
\label{alg:corpus_zipf}
\begin{algorithmic}[1]
\Require Number of tokens $N_{\text{tokens}}$
\Require MCWM parameters (morpheme inventories, activation probabilities, categorical distributions)
\Require SLF survival function $\phi(w)$
\Ensure Empirical rank--frequency curve $\{(r, f(r))\}$

\State Initialize an empty dictionary \texttt{counts} mapping word types to integer counts

\For{$t = 1$ to $N_{\text{tokens}}$}
    \State Generate a candidate word $(w, L(w))$ using Algorithm~\ref{alg:single_word}
    \State Compute survival indicator $S(w)$ using Algorithm~\ref{alg:slf}
    \If{$S(w) = 1$}
        \State Update \texttt{counts[$w$]} $\gets$ \texttt{counts[$w$]} + 1
    \EndIf
\EndFor

\State Extract all word types with positive counts: $\{w_1, \dots, w_K\}$
\State Compute empirical frequencies
    \[
        f(w_i) = \frac{\texttt{counts}[w_i]}{\sum_{j=1}^K \texttt{counts}[w_j]}
        \quad \text{for } i = 1,\dots,K.
    \]

\State Sort $\{w_i\}$ in decreasing order of $f(w_i)$ to obtain ranks $r = 1,2,\dots,K$
\State Define the rank--frequency function $f(r)$ by assigning $f(1)$ to the most frequent word, $f(2)$ to the second most frequent, and so on

\State \Return $\{(r, f(r)) : r = 1,\dots,K\}$
\end{algorithmic}
\end{algorithm}

This algorithm corresponds directly to the simulation procedures in the main
paper. By varying:
\begin{itemize}
    \item the activation probabilities of the slots,
    \item the morpheme inventories and their length distributions,
    \item the survival function $\phi(w)$,
\end{itemize}
one can explore how the shape of the empirical Zipf curve changes, and verify
the robustness of the theoretical results with respect to model parameters.

\subsection{Remarks on implementation}

In practical implementations, several optimizations are useful:
\begin{itemize}
    \item Caching morpheme lengths to avoid repeated length computations.
    \item Precomputing templates and reusing them to study the effect of
    different lexical filters.
    \item Using efficient hash maps or dictionaries for counting word types.
    \item Parallelizing the token generation loop when $N_{\text{tokens}}$ is large.
\end{itemize}

These details do not affect the theoretical conclusions of the paper, but they
make it easier to reproduce the figures and to extend the model to larger
synthetic corpora or to more complex morphological scenarios.

%% file: sections/references.tex
=